\documentclass[11pt,a4paper]{article}
\pdfoutput=1
\usepackage[utf8]{inputenc}
\usepackage{amsmath}
\usepackage{amsfonts}
\usepackage{array}
\usepackage{bbm}
\usepackage{amssymb}
\usepackage{cite}
\usepackage{graphicx}
\usepackage{enumitem}
\usepackage{appendix}
\usepackage{amsthm}
\usepackage{subcaption}
\usepackage{multirow}
\usepackage[pdftex,pdfauthor={Ioannis Z. Koukoutsidis},pdftitle={A fluid reservoir model for the Age of Information through energy-harvesting transmitters}]{hyperref}
\author{Ioannis Z. Koukoutsidis}
\title{A fluid reservoir model for the Age of Information through energy-harvesting transmitters\footnote{Paper presented at the 2021 International Symposium on Performance Evaluation of Computer and Telecommunication Systems (SPECTS 2021). \copyright The Society for Modeling \& Simulation International (SCS) 2021}}
\newtheorem{lemma}{Lemma}

\theoremstyle{definition}

\begin{document}
\maketitle
\begin{abstract}
We apply a fluid-reservoir model to study the Age-of-Information (AoI) of update packets through energy-harvesting transmitters. The model is closer to how energy is stored and depleted in reality, and can reveal the system behavior for different settings of packet arrival rates, service rates, and energy charging and depletion rates. We present detailed results for both finite and infinite transmitter buffers and an infinite energy reservoir, and some indicative results for a finite reservoir. The results are derived for the mean AoI in the case of an infinite transmitter buffer and an infinite reservoir, and for the mean peak AoI for the remaining cases. The results show that, similar to a system without energy constraints, the transmitter buffer should be kept to a minimum in order to avoid queueing delays and maintain freshness of updates. Furthermore, a high update packet rate is only helpful in energy-rich regimes, whereas in energy-poor regimes more frequent updates deplete the energy reservoir and result in higher AoI values.
\end{abstract}
\section{Introduction}
\label{sec:Intro}
The Age of Information (AoI), or simply age, refers to the freshness of information from a remote source, until it is received at its intended destination. In packet networks, we consider that information updates are
received in packets, and at each time instant $t$, the destination observes an age $\text{AoI}(t) = t - u(t)$, where $u(t)$ is the generation time of the last packet received. 

The original definition in \cite{kaul2012real} considered the \textit{mean} AoI, which for a stationary and ergodic process converges to the average age seen by the destination at a random instant in time. Another measure is the \textit{mean peak} AoI, introduced in \cite{costa2014age}, which converges to the average of the maximum age values, prior to the packets being received at the destination.

The difference between the two metrics is illustrated in Fig.~\ref{sawtooth}. We consider that packets arrive at a queue for transmission.  $S_i$ is the sojourn time of the $i$-th packet received at the destination, $\alpha_i$ its arrival instant and $d_i$ its departure instant ($S_i=\alpha_i-d_i$). For simplicity, assume that the packets arrive with zero age, so that the arrival times equal the generation times; additionally, that there is no delay from departure until packets are received at the destination. 
The age process is depicted as the sawtooth curve, and is right-continuous with left limits. The solid black circles correspond to the sojourn times, while the left limits (solid white circles) average to the mean peak age of the update process.

\begin{figure}[!htb]
	\centering
	\includegraphics[scale=0.9]{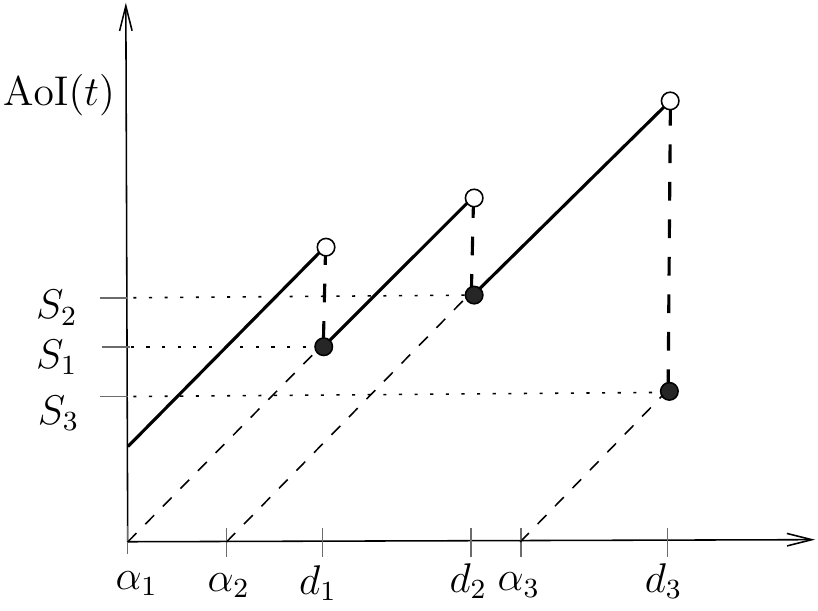}
	\caption{Sample path start of the age update process\label{sawtooth}}
\end{figure}

The mean AoI equals
\begin{displaymath}
	\mathbb{E}[\text{AoI}]=\lim_{t\to\infty}\frac{1}{t}\int_{0}^{t}\text{AoI}(t)dt\;.
\end{displaymath}
Denote the interarrival interval $A_i\equiv \alpha_{i+1}-\alpha_{i}$ and the interdeparture interval by $D_i\equiv d_{i+1}-d_{i}$. As can be seen from the figure, the peak age of update packet $i$, $\text{AoI}_{i,peak}$, equals $\text{AoI}_{i,peak}=A_{i-1}+S_{i}=D_{i-1}+S_{i-1}$. For a stationary and ergodic process, taking expectations leads to
\begin{equation}
	\mathbb{E}[\text{AoI}_{peak}]=\mathbb{E}[A]+\mathbb{E}[S]=\mathbb{E}[D]+\mathbb{E}[S]\;.
	\label{mean_peak_age}
\end{equation}
Clearly, the mean peak AoI is also an upper bound to the mean AoI. 
There also exist other age-related metrics, but only these two will be used in this paper.\footnote{Examples are the {\it New} AoI defined in \cite{kesidis2020distribution}, as well as age metrics with non-linear penalty functions to represent differently the value of information at the receiver as data gets older \cite{sun2017update}.}.

The AoI -- in its various forms -- can serve as a metric in numerous applications, where we are interested in the freshness of the received information. For example, in the Internet of Things, where sensor devices can transmit updates of environmental parameters, or the values of technical parameters such as location and velocity in autonomous vehicles; in the storage of data in computer systems, where we are interested in the freshness of imformation in the cache memory, or in robotics and control systems, where the fast feedback of information plays a prominent role. An extensive survey of the AoI metric and its applications can be found in \cite{yates2020age}.

A large number of works have considered energy-constrained updating, where the ability to make an update is constrained by 
the energy of the transmitter. Attempts to calculate the impact of energy constraints have so far considered the simplifying assumption of Poisson arrivals of discrete ``energy units'' --- one unit assumed to provide enough energy to serve one packet (see Section~\ref{sec:rel_work} for a description of related work). 

In this paper, we depart from this simplifying assumption and apply the fluid-reservoir model presented in \cite{adan1998analysis} to derive results for the mean AoI and the mean peak AoI for an energy-constrained transmitter. In this model, packets are queued at a server for transmission and energy is considered as a fluid commodity, which accumulates when the server is empty and depletes when the server is busy --- at different but constant rates. The model is closer to reality, and avoids certain artifacts of Poisson energy arrival models (see Section~\ref{sec:rel_work}).

The model is presented in Section~\ref{sec:model}. Throughout the paper, a FCFS (First-Come First-Served) service discipline is assumed. For the sake of comparison, a transmitter queue that is regulated by the energy in the reservoir is termed as a ``regulated queue'', in contrast to an ``unregulated queue'', where there are no energy constraints. We will follow the structure of \cite{adan1998analysis}, which examined different cases, depending on the size of the waiting room in the server and the size of the reservoir. Results for the mean AoI are derived only for the case of infinite waiting room in the server and an infinite reservoir, as it is the only case where a closed-form expression for the sojourn time distribution can be derived (which is a requirement for calculating the mean AoI). These are presented in Section~\ref{sec:inf-inf}. For the case of finite waiting room and an infinite reservoir, Section~\ref{sec:fin-inf} presents results for the mean peak AoI, based on the stationary distribution of the system. We derive closed-form formulas for the mean AoI in the case of an M/M/1 queue, and the mean peak AoI in the case of an M/M/1/1 queue, both under the assumption of an infinite reservoir. Apart from that, we conduct an extensive performance evaluation, comparing the finite and infinite buffer cases, as well as performances for the regulated and unregulated queues. Finally, in Section~\ref{sec:inf-fin} we present some indicative results for the mean peak AoI in the case of a finite capacity reservoir, based on values of the mean sojourn time derived by simulation in \cite{adan1998analysis}. The major conclusions are presented in Section~\ref{sec:conclusions}, along with a description of issues that remain open.
\section{Related Work}
\label{sec:rel_work}
Previous research in energy-constrained updating can be grouped in two categories: a) non-queueing settings, which consider that transmitters can generate updates at will, with zero delay from the time of packet generation till packet transmission, and b) queueing settings, where the updates arrive exogenously by another process (not controlled by the transmitter), and are queued for transmission.  Although this work belongs to the second category, it is worth visiting some results derived in non-queueing settings.

In such (non-queueing) settings, some sort of scheduling problem is usually considered, in which one must choose the optimal points in time in which to send the updates. Most prominent references are \cite{wu2017optimal} and \cite{arafa2019age}, where energy is assumed to arrive in units in a reservoir as a Poisson process with a rate of one arrival per time unit, and each update packet consumes one such unit of energy for transmission. In the case of infinite reservoir capacity, it was shown in \cite{wu2017optimal} that the optimal strategy is to send updates at each time unit, provided there is energy available; this is no surprise, since this particular setting corresponds to an M/M/1 queue with load equal to one, and hence in the long run the reservoir will grow to infinity. In the case of a finite-capacity reservoir, it was shown in \cite{arafa2019age} that the transmitter should send packets more frequently, when it has relatively higher energy available, and less frequently, when it has relatively lower energy available.

In queueing settings, an initial high-impact work was \cite{farazi2018average}. In the model in that paper, new updates are considered to enter service if the server is idle and has sufficient energy to service the packet; whereas if the server is busy or there is no sufficient energy, the packet is dropped (essentially corresponding to an M/M/1/1 queue with the addition of an energy container). The authors considered both cases where the container \emph{is} able to harvest energy when the server is busy, and \emph{is not} (showing the decrease in AoI in the first case).
In \cite{farazi2018age}, the authors extended this work for preemption in service, with the assumption that the unit of energy ``assigned'' to the update of a packet is lost, if the packet is preempted in service. They showed that preemption in service can further decrease the average age, but only in ``energy-rich'' operating regimes (otherwise it can lead to waste of energy and inability to serve subsequent packets with fresher information). Similarly to their previous work, they also showed that a lower AoI is achieved when the server is also able to harvest energy while it is busy.
Finally, the work in \cite{zheng2019closed} extended the results to queues of arbitrary length and LCFS (Last-Come First-Served) policies, and non-linear penalty functions. The authors derived closed-form expressions for the mean AoI and other penalty functions under the FCFS and LCFS when the service time is negligible. When the service time is not negligible, they developed a method based on QBD (Quasi-Birth Death) processes to numerically find the mean peak AoI.

The assumption of negligible or zero service time is reminiscent of classical leaky-bucket schemes, which are concerned with the process in which a packet obtains ``access'' or ``permission'' for transmission, and the generated rate of incoming packets \cite[pp.~511-515]{bertsekas1992robert}. By contrast, the model used here can incorporate both transmit buffers of arbitrary sizes (thus allowing the evaluation of M/M/1/N schemes), and non-zero service time at the transmitter.

Another major difference with the works in \cite{farazi2018average,farazi2018age,zheng2019closed} lies in the energy model used. Similarly to \cite{wu2017optimal} and \cite{arafa2019age}, \cite{farazi2018average,farazi2018age,zheng2019closed} assumed that energy arrives in units and the service of a packet always consumes one unit of energy. Both update packets and energy units arrive as Poisson processes. This setting leads to an interplay between the arrival rates of the two processes and in some cases in strange artifacts, such as the invariance of the AoI if we interchange the packet and energy arrival rates (see \cite{farazi2018average,farazi2018age}). By contrast, the model used here considers constant charging and depletion rates, which is closer to reality. Indeed, in the case of solar battery charging (which is a likely scenario for IoT devices), it has been shown that charging is done at an almost constant rate; despite drops in the incoming charging current (due, e.g. to a more clouded sky), the charging voltage remains constant for periods of several minutes or even hours (see, e.g. \cite{suresh2014efficient,boico2007solar}). Therefore, although the assumption of Poisson energy arrivals simplifies the analysis, it implies significant randomness in the charging rate which does not reflect reality. As far as energy is concerned, the model used here is more realistic, and, despite its limitations, it can reveal the system behavior for different settings of update packet arrivals, service rates, and energy charging and depletion rates.
\section{The Model}
\label{sec:model}
We shortly describe the model in \cite{adan1998analysis}. We consider that packets arrive from an external source to a transmitter according to a Poisson process with rate $\lambda$. The transmitting process is modeled as a FCFS queue; both infinite and finite queues will be examined. Service is regulated by a fluid reservoir, which determines the amount of energy which is available for transmission. Energy is treated as a fluid commodity, which fills during idle periods of the server at a constant rate $r_{+}$, and depletes during busy periods at a constant rate $r_{-}$, as long as the reservoir is non-empty (Fig.~\ref{model}).
When the reservoir is non-empty, the service time of update packets is exponentially distributed with rate $\mu_1$, and when the reservoir is empty, it is exponentially distributed with rate $\mu_2\leq \mu_1$. 
\begin{figure}[!htb]
	\centering
	\includegraphics[scale=0.9]{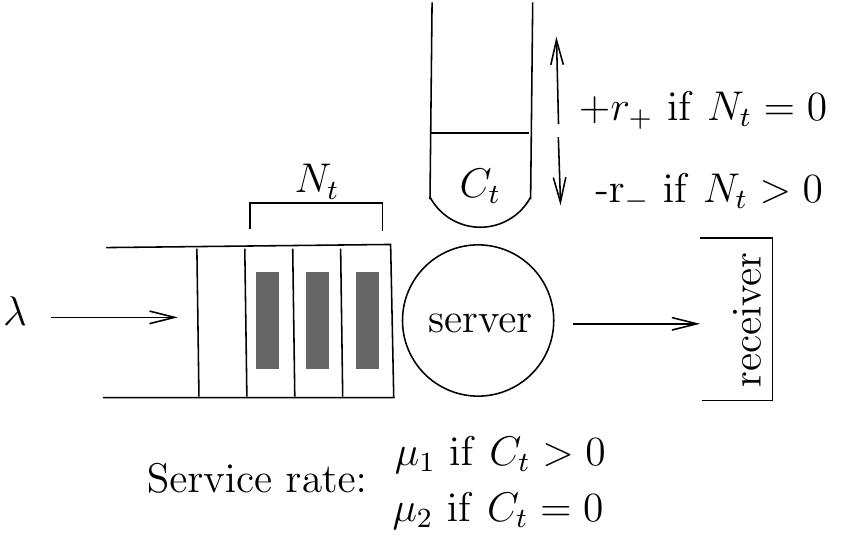}
	\caption{The fluid-reservoir model ($N_t$ is the number of packets in the system and $C_t$ the amount of fluid in the reservoir at time $t$)\label{model}}
\end{figure}

Obviously, when $\mu_1=\mu_2$ the reservoir does not effect the service process, and therefore the model is the same as that of a simple (unregulated) queue. When $\mu_2< \mu_1$, the impact of the energy constraint is that it delays the transmitted packets, so that they arrive later at the receiver. Our goal is to find how that impacts on the AoI at the receiver.

It is emphasized that when the reservoir is empty, the service does not stop, but just slows down. This is done for analytical tractability, but could also model situations when the transmitter goes into battery-saving mode (as long as the energy threshold for such a mode is calibrated so that the transmitter does not go completely out of energy for a relatively long period). The model can also include the case where the receiver is able to accummulate energy during busy hours (as was examined in \cite{farazi2018average,farazi2018age}), by simply considering lower depletion rates.\footnote{The case of a negative depletion rate is not interesting, since the transmitter would never run out of energy.}
\section{Infinite Waiting Room and Infinite Reservoir}
\label{sec:inf-inf}
For a FCFS M/M/1 queue and an infinite reservoir, exact expressions were derived in \cite{adan1998analysis} for the distributions of the stationary sojourn time $S$ and the stationary waiting time $W$ of a customer. 
The distribution of the stationary sojourn time is:
\begin{equation}
	\text{Pr}\{S>s\}=\zeta e^{-\lambda\sigma^{-1}(1-
		\sigma)s}+(1-\zeta)e^{-(\mu_2-\lambda)s}, \quad s\geq 0
	\label{soj_time_distr}
\end{equation}

where
\vspace{6pt}

$\zeta\equiv\displaystyle\frac{(\mu_1-\mu_2)(\mu_2-\lambda)\sigma}{(\lambda-\mu_2\sigma)(\mu_2-\mu_1\sigma)}<1\quad\refstepcounter{equation}(\theequation)\label{zeta}$ \quad and \quad $\sigma=\displaystyle\frac{r_+}{r_+ + r_-}\quad\refstepcounter{equation}(\theequation)\label{sigma}$\;.

\vspace{10pt}
The stationary waiting time distribution is:
\begin{equation}
	\text{Pr}\{W>s\}=\zeta\sigma e^{-\lambda\sigma^{-1}(1-
		\sigma)s}+\eta e^{-(\mu_2-\lambda)s}, \quad s\geq 0
\end{equation}

where $\zeta$ is as above and $\eta\equiv\displaystyle\frac{\lambda(\lambda-\mu_1\sigma)(1-\sigma)}{(\lambda-\mu_2\sigma)(\mu_2-\mu_2\sigma)}\quad\refstepcounter{equation}(\theequation)\label{eta}$\;.

\vspace{10pt}
Throughout we assume that the stability condition 
\begin{equation}
	\sigma<\lambda/\mu_1\leq \lambda/\mu_2<1
	\label{stability_condition}
\end{equation}
holds, so that the stationary distributions exist (the condition $\sigma<\lambda/\mu_1$ is for the stability of the fluid reservoir, whereas the conditions $\lambda/\mu_1\leq \lambda/\mu_2<1$ are for the stability of the infinite queue \cite{adan1998analysis}).

From the above distributions we can also calculate the expected values of the sojourn time and waiting time of a customer in the system:
\begin{align}
	\mathbb{E}[S]=\int_0^\infty \text{Pr}\{S>s\}ds=\frac{\zeta}{\lambda\sigma^{-1}(1-\sigma)}+\frac{1-\zeta}{\mu_2-\lambda}
	\label{exp_soj_time}
\end{align}
and 
\begin{align}
	\mathbb{E}[W]=\int_0^\infty \text{Pr}\{W>s\}ds=\frac{\zeta\sigma}{\lambda\sigma^{-1}(1-\sigma)}+\frac{\eta}{\mu_2-\lambda}\;.
	\label{exp_wait_time}
\end{align}

We then proceed to derive a closed-form expression for the mean AoI. 
From the analysis of the AoI for the M/M/1 FCFS queue \cite{kaul2012real}, we know that
\begin{equation}
	\mathbb{E}{[\text{AoI}]}=\frac{\mathbb{E}[S A]}{\mathbb{E}[A]}+\frac{\mathbb{E}[A^2]}{2\mathbb{E}[A]}\;,
	\label{EAoI}
\end{equation}
where $A$ is the interarrival interval and $S$ is the sojourn time of update packets in the system.
Remember that the interarrival distribution is assumed to be exponential with rate $\lambda$, so $\mathbb{E}[A]=1/\lambda$
and $\mathbb{E}[A^2]=2/\lambda^2$.

Splitting the sojourn time into waiting time $W$ and service time $X$, we have:
\begin{equation}
	\mathbb{E}[SA]=\mathbb{E}[(W+X)A]=\mathbb{E}[WA]+\mathbb{E}[X]\mathbb{E}[A]\label{ESA}
\end{equation}
(from independence of service time and interarrival time).
The expected service time can be calculated from (\ref{exp_soj_time}), (\ref{exp_wait_time}) as:
\begin{align}
	\mathbb{E}[X]&=\mathbb{E}[S]-\mathbb{E}[W]=\frac{\zeta\sigma}{\lambda}+\frac{1-\zeta-\eta}{\mu_2-\lambda} \nonumber\\
	&=\frac{\zeta\sigma}{\lambda}+\frac{(1-\sigma)(\lambda\mu_2-\mu_1\mu_2\sigma-\lambda^2+\lambda\mu_1\sigma)}{(\mu_2-\lambda)(\lambda-\mu_2\sigma)(\mu_2-\mu_1\sigma)}\;.\label{EX}
\end{align}
The expected value of the product of waiting and interarrival times is calculated as follows:
\begin{align}
	\mathbb{E}[WA]&=\mathbb{E}\left[\mathbb{E}[WA|A=a]\right]=\mathbb{E}\left[a\mathbb{E}[(S-a)^+]\right]\nonumber\\
	&=\int_{0}^{\infty}\int_{a}^{\infty}a(t-a)f_S(t)f_A(a)dt da\nonumber\;,
\end{align}
where $f_S$, $f_A$ are the probability density functions of the sojourn time and the interarrival time, respectively.

From (\ref{soj_time_distr}), we can easily derive that
\begin{align*}
	f_S(t)=&\zeta \lambda\sigma^{-1}(1-\sigma)e^{-\lambda\sigma^{-1}(1-\sigma)t}\\
	&+(1-\zeta)(\mu_2-\lambda)e^{-(\mu_2-\lambda)t}\;.
\end{align*}
In addition, $f_A(a)=\lambda e^{-\lambda a}$. We therefore have:
\begin{align}
	&\mathbb{E}[WA]=\nonumber\\
	&\int_0^\infty a \lambda e^{-\lambda a} \int_{a}^{\infty} (t-a) \left[\zeta \lambda\sigma^{-1}(1-\sigma)e^{-\lambda\sigma^{-1}(1-\sigma)t}\right]dt da\nonumber \\
	&+\int_0^\infty a \lambda e^{-\lambda a} \int_{a}^{\infty} (t-a) \left[(1-\zeta)(\mu_2-\lambda)e^{-(\mu_2-\lambda)t}\right]dt da\nonumber\\
	=&\int_{0}^{\infty} a \lambda e^{-\lambda a} \zeta \frac{e^{-\lambda\sigma^{-1}(1-\sigma)a}}{\lambda\sigma^{-1}(1-\sigma)}da\nonumber\\
	&+\int_0^\infty a \lambda e^{-\lambda a}(1-\zeta)\frac{e^{-(\mu_2-\lambda)a}}{\mu_2-\lambda}da\nonumber\\
	=&\frac{\zeta\sigma}{1-\sigma}\frac{1}{(\lambda+\lambda\sigma^{-1}(1-\sigma))^2}+\frac{\lambda(1-\zeta)}{\mu_2-\lambda}\frac{1}{\mu_2^2}\;.\label{EWA}
\end{align}

Substituting (\ref{EWA}) and (\ref{EX}) in (\ref{ESA}), (\ref{EAoI}) finally becomes:
\begin{align}
	\mathbb{E}{[\text{AoI}]}=&\frac{\lambda\zeta\sigma}{(1-\sigma)(\lambda+\lambda\sigma^{-1}(1-\sigma))^2}+\frac{\lambda^2(1-\zeta)}{(\mu_2-\lambda)\mu_2^2}\nonumber\\
	&+\frac{\zeta\sigma}{\lambda}+\frac{(1-\sigma)(\lambda\mu_2-\mu_1\mu_2\sigma-\lambda^2+\lambda\mu_1\sigma)}{(\mu_2-\lambda)(\lambda-\mu_2\sigma)(\mu_2-\mu_1\sigma)}+\frac{1}{\lambda}\;.\label{EAoI2}
\end{align}
Note that for $\mu_1=\mu_2=\mu$, $\zeta=0$ and (\ref{EAoI2}) becomes 
\begin{displaymath}
	\frac{\lambda^2}{\mu^2(\mu-\lambda)}+\frac{1}{\mu}+\frac{1}{\lambda}\;,
\end{displaymath}
which is the expected value of the AoI of a single server M/M/1 FCFS queue, as originally derived in \cite{kaul2012real}.

We proceed to evaluate the performance of this system. Our aim is to see how the mean AoI behaves for different values of arrival and service rates, and different fill and depletion rates of the fluid reservoir (e.g. when the reservoir fills faster than it depletes, and vice-versa).

Our basis for comparison of the results will be the unregulated FCFS M/M/1 queue. Obviously, since the reservoir introduces delay, we expect the AoI to always be greater in the regulated queue. Another aspect relates to the major conclusion in \cite{kaul2012real}, namely that in the FCFS M/M/1 queue there is an optimal intermediate load value, for which the $\mathbb{E}$[AoI] is minimal ($\lambda\approx 0.53 \mu$). That is, if we want to minimize the mean AoI at the receiver, it is not optimal to send packets neither at a slow rate, nor at a fast rate. We want to see if this conclusion still holds for the regulated queue, and if yes what is the required load in comparison to the unregulated queue.

One major constraint is that the parameter space is smaller in comparison to the unregulated queue, since we have the additional stability constraint for the fluid reservoir. Equation (\ref{stability_condition}) says that, in order to have stability of the reservoir, we need to send packets above a certain arrival rate (for a certain service rate), so that the reservoir becomes empty with non-zero probability. This significantly constraints the allowed values of $\lambda$; since we know that it is not optimal to send update packets at a fast rate, the question is whether the optimal arrival rate is at the lower edge or inside the stability region.

\begin{figure*}[!htb]  
	\begin{subfigure}{0.5\textwidth}
		\includegraphics[width=\linewidth]{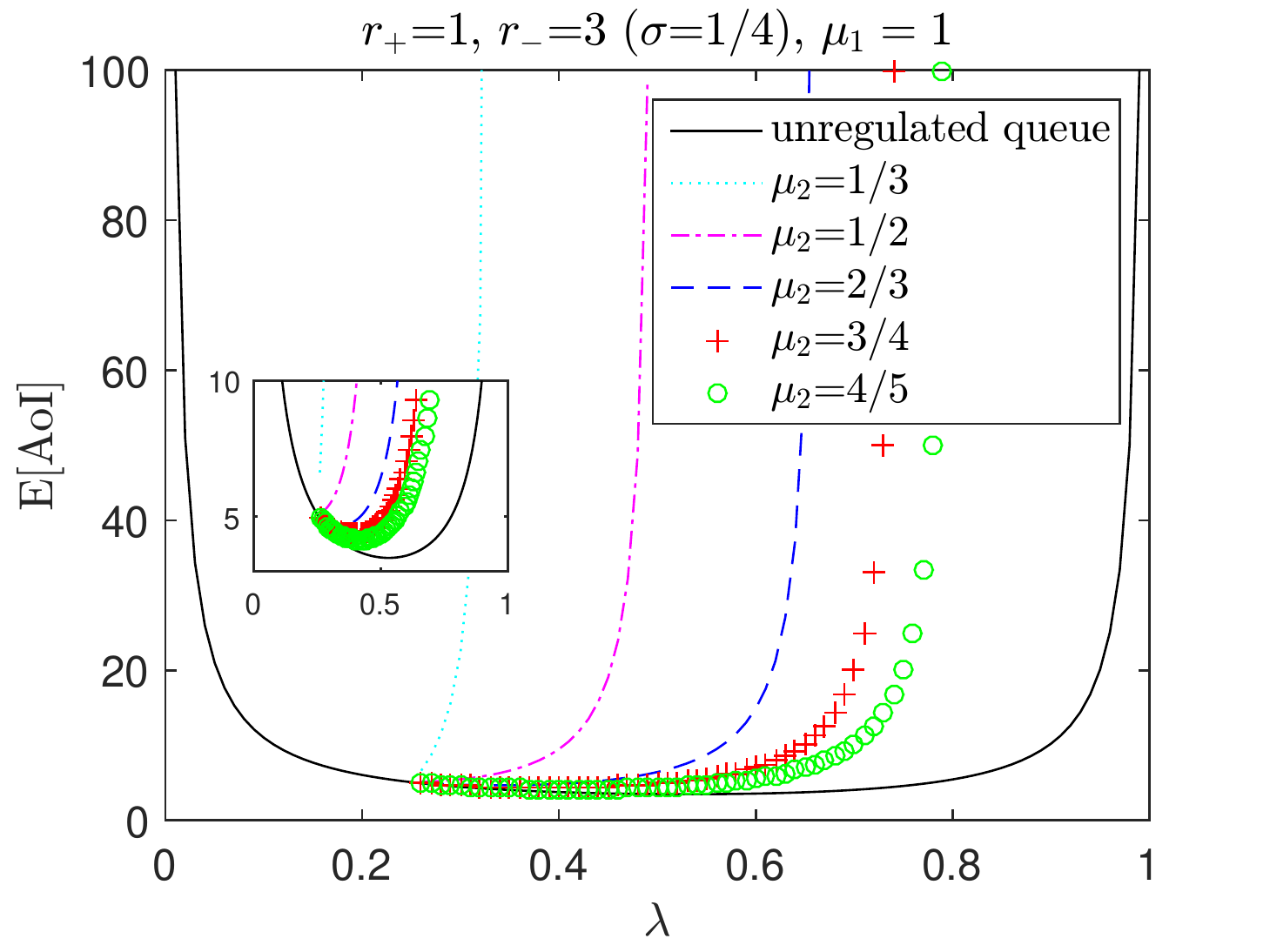}
		\caption{}
		\label{plot_sc_1_4}
	\end{subfigure}
	\hfill 
	\begin{subfigure}{0.5\textwidth}
		\includegraphics[width=\linewidth]{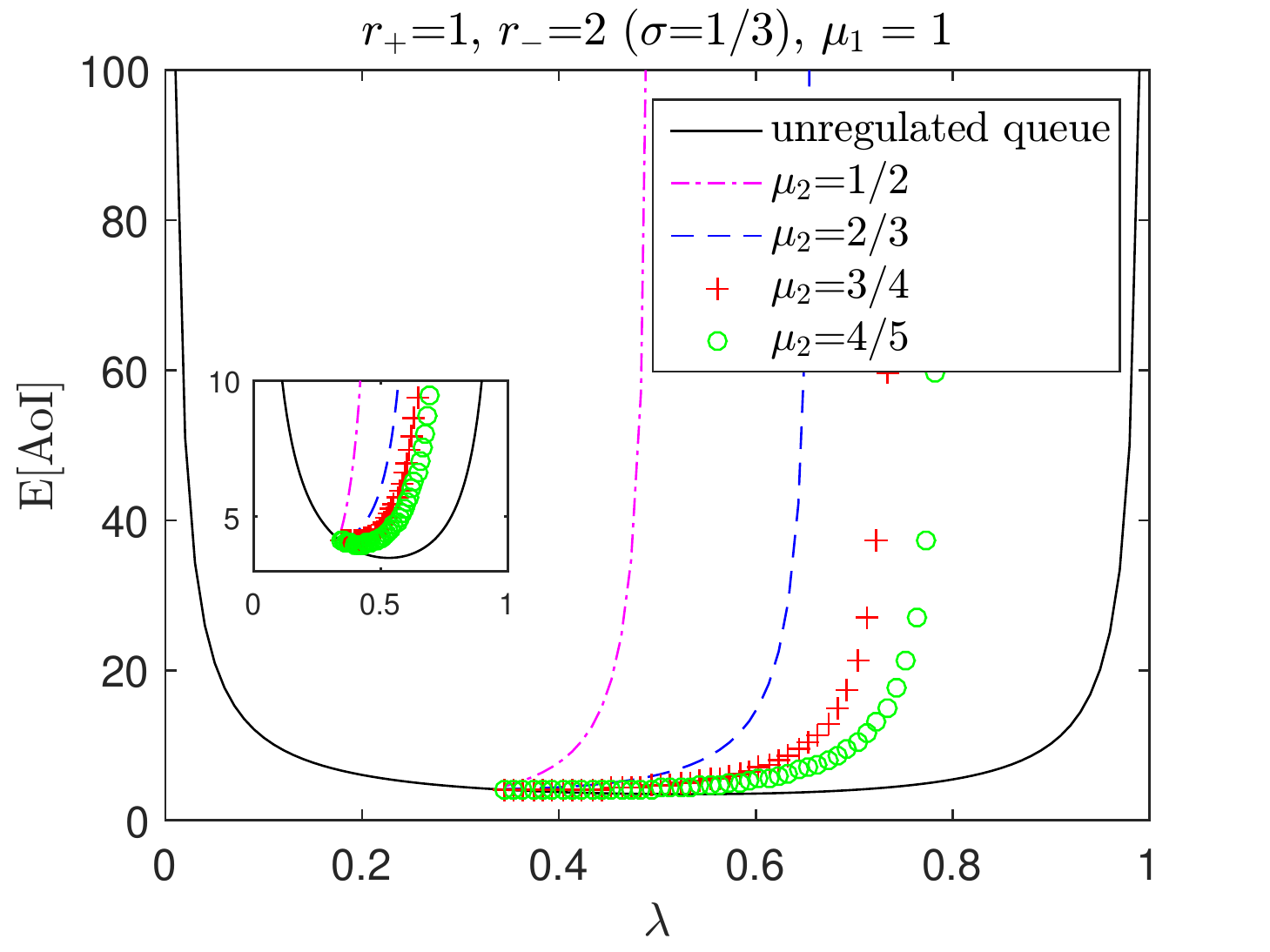}
		\caption{}
		\label{plot_sc_1_3}
	\end{subfigure}
	\begin{subfigure}{0.5\textwidth}
		\includegraphics[width=\linewidth]{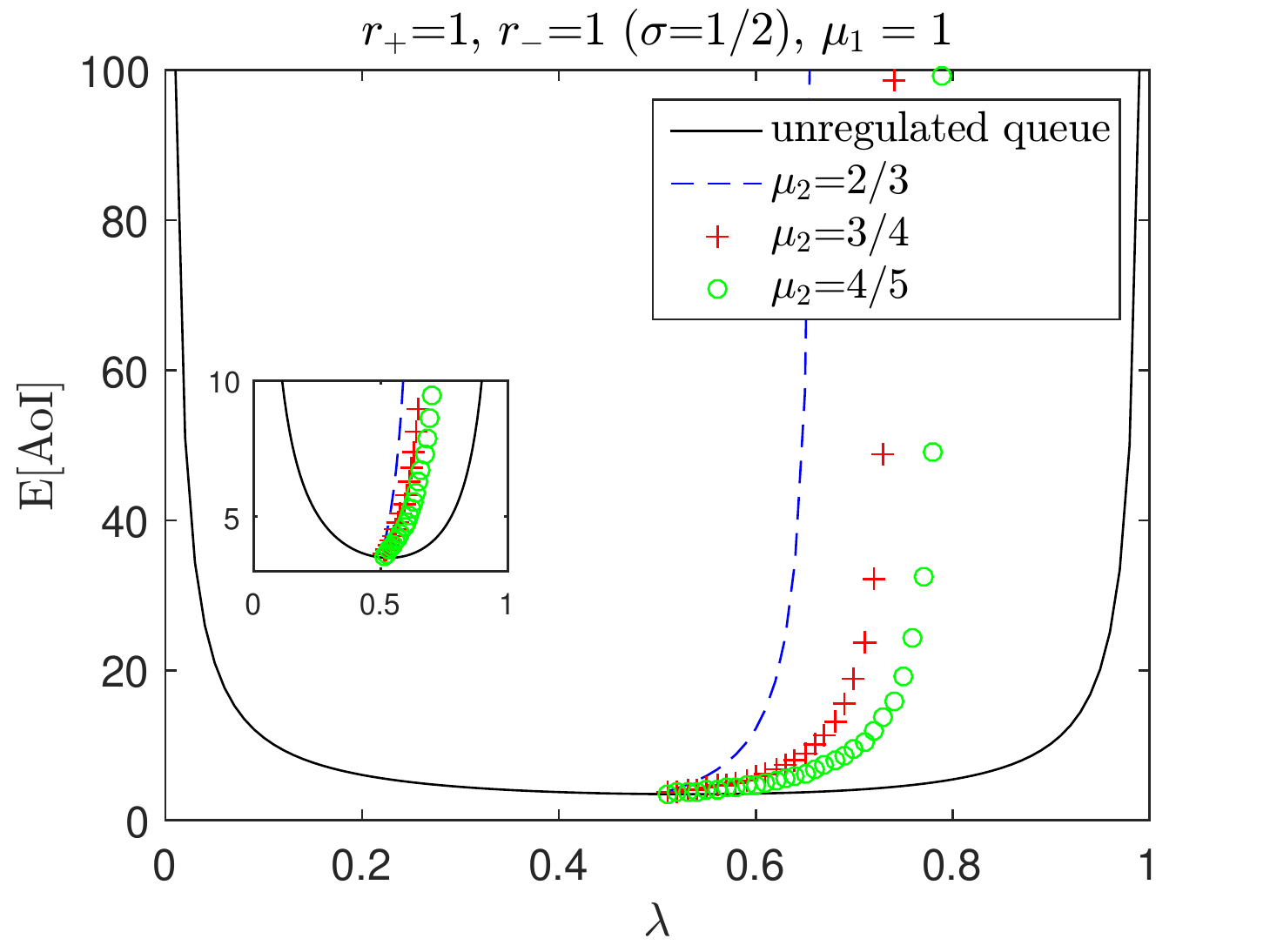}
		\caption{}
		\label{plot_sc_1_2}
	\end{subfigure}
	\hfill 
	\begin{subfigure}{0.5\textwidth}
		\includegraphics[width=\linewidth]{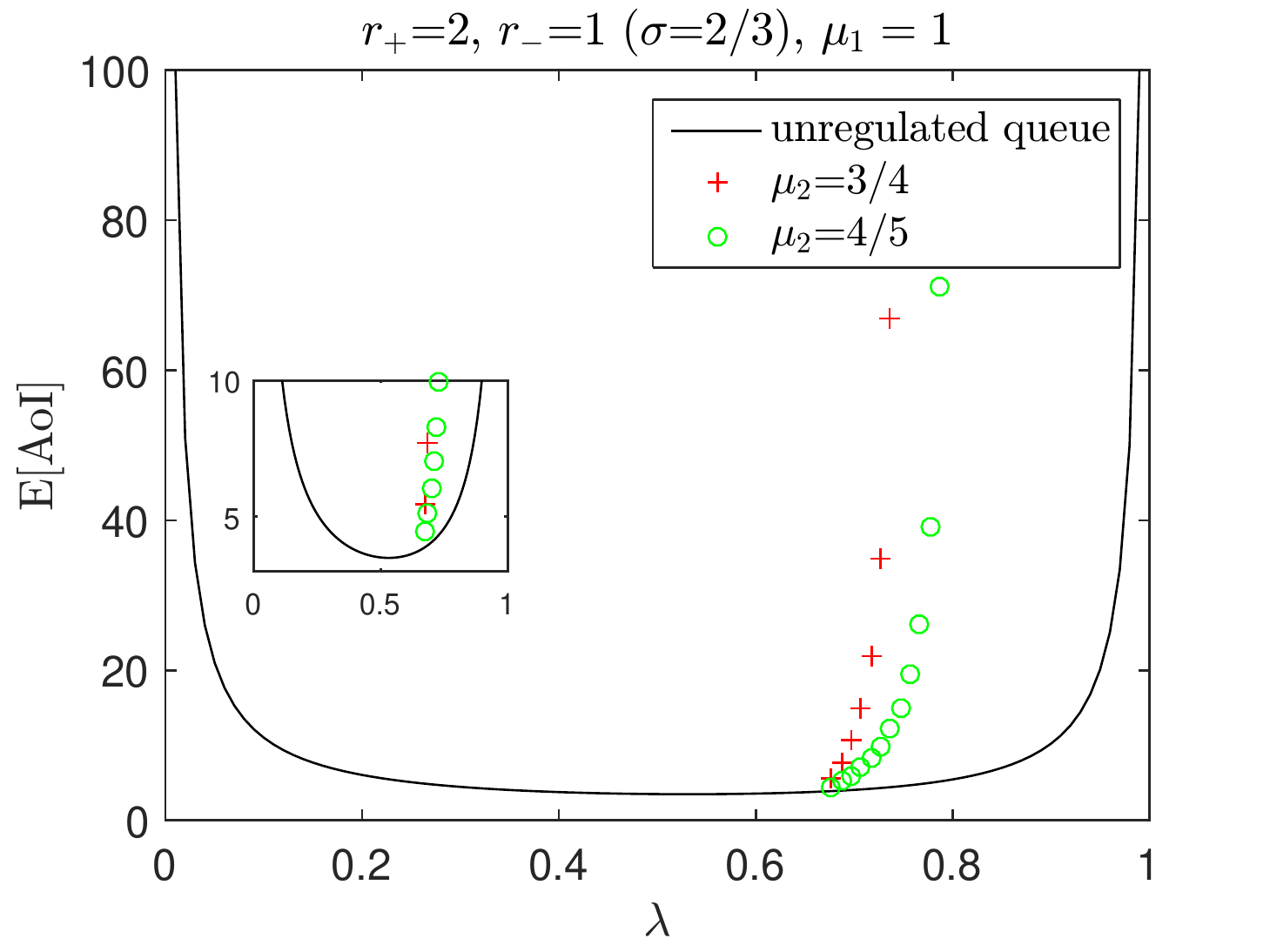}
		\caption{}
		\label{plot_sc_2_3}
	\end{subfigure}
	\caption{$\mathbb{E}[\text{AoI}]$ for infinite waiting room and infinite reservoir} 
	\label{plots_single_class}
\end{figure*}

Results are shown in Fig.~\ref{plots_single_class}. In all results, we keep the service rate $\mu_1$ (the rate under a non-empty 
reservoir) equal to unity, and vary the remaining parameters. Each subfigure shows a different pair of ($r+$, $r_-$) values. 
A larger depletion vs. fill rate leads to ``energy-poorer'' regimes. Note also that a larger depletion rate allows to send packets at a smaller rate, whereas the arrival rate must be faster, if the fill rate is faster (in order to consume credits and maintain stability). 

If the depletion rate is faster than the fill rate (Fig.~\ref{plot_sc_1_4} and \ref{plot_sc_1_3}), then the optimal
arrival rate is within the stability region, if in addition the service rate $\mu_2$ (the service rate under an empty reservoir) is 
not too low. This can be seen in the zoomed-in inset plots of Figs.~\ref{plot_sc_1_4}, \ref{plot_sc_1_3}. In Fig.~\ref{plot_sc_1_4}, for $\mu_2=2/3$, the optimal arrival rate is about 0.34, whereas for $\mu_2=4/5$, it is about 0.42. These are smaller than the optimal value in the unregulated case ($\lambda\approx0.53$). Generally, in order to achieve a small age, the energy constraints require to decrease the update packet arrival rate further than in the unregulated case, and it is near-optimal to send update packets with an arrival rate that is as small as possible within the stability region.

We also note that as the service rate $\mu_2$ becomes smaller, the mean AoI increases very fast,
even for moderate values of the arrival rate. We then need to decrease the arrival rate even further, in order to achieve a small age.  For example, in Fig.~\ref{plot_sc_1_4}, for $\mu_2=1/3$ an arrival rate smaller than 0.26 is needed to keep the mean AoI smaller than twice the optimal mean AoI in the unregulated queue. These results imply that the length of the update queue (which increases for smaller $\mu_2$) is still the most important parameter, and must be kept small to achieve a small age. They also provide a foretaste of the gains that can be achieved by small buffers, which is studied in the next section.
\section{Finite Waiting Room and Infinite Reservoir}
\label{sec:fin-inf}
Here we introduce a constraint $N$ on the number of the update packets that can be stored in the transmitter's buffer (including the packet in service), and proceed to calculate the mean peak age, denoted by $\mathbb{E}[\text{AoI}_{peak,N}]$. From Section~\ref{sec:Intro}, this equals $\mathbb{E}[A]+\mathbb{E}[S]$, where $S$ is the sojourn time and $A$ is the interarrival time between update packets that are successfully received at the destination (i.e. excluding the packets that are blocked). 
For a FCFS discipline and a Poisson arrival process, $\mathbb{E}[A]=1/\lambda(1-P_{N})$, where $P_{N}$ is the blocking probability, or the probability that an arriving packet sees the buffer full. The difficult part is the calculation of the mean sojourn time. For this we rely on the following lemma:
\begin{lemma}
	\label{lemma:finite_buff}
	In the model of Sect.~\ref{sec:model}, the limiting probability of the number of packets in the system, $p_i$, $i\in \mathcal{N}=\{0,1,\dots,N\}$ can be found by solving the system of $2N+2$ linearly independent equations
	\begin{equation}
		\begin{aligned}
			\mu_1 y_1&=(\lambda+r_ +\xi_0^{(N)})y_0\;, \\
			\mu_1 y_{i+1}&=(\lambda+\mu_1-r_ -\xi_0^{(N)})y_i-\lambda y_{i-1}\;, i\in \mathcal{N}\backslash\{0,N\}\;,\\
			p_0&=y_0\;,\\		
			\mu_2 p_{i+1}&=\lambda p_i-(\mu_1-\mu_2)y_{i+1}, i\in\mathcal{N}\backslash\{N\}\;,\\
			\sum_{i\in\mathcal{N}}p_1&=1
		\end{aligned}
		\label{sys_eq_fin_buff}
	\end{equation}
	where $\xi_0^{(N)}$ is the unique negative zero of the polynomial $P_N(x)$, $N=1,2,\dots$, defined by the recurrence relations
	\begin{subequations}
		\begin{align}
			P_0(x)&=1\;,\nonumber\\
			P_1(x)&=x+\frac{\lambda}{r_+}-\frac{\mu_1}{r_-}\;,\label{P_1(x)}\\
			P_N(x)&=\left(x+\frac{\lambda+\mu_1}{r_-}\right)P_{N-1}(x)-\frac{\lambda\mu_1}{r_-^2}P_{N-2}(x),\\
			&\qquad\qquad\qquad\qquad\qquad\qquad N=2,3,\dots\nonumber.
		\end{align} 
	\end{subequations}
\end{lemma}
\begin{proof}
	See \cite[Section 2]{adan1998analysis}.
\end{proof}
Since 
\begin{displaymath}
	\left\{1+\frac{\lambda}{\mu_1}+\left(\frac{\lambda}{\mu_1}\right)^2+\cdots+\left(\frac{\lambda}{\mu_1}\right)^N \right\}^{-1}
\end{displaymath}
is the stationary probability of the server being idle in a M/M/1/N system with arrival rate $\lambda$ and service rate $\mu_1$, for stability of the energy reservoir we must have
\begin{align}
	\frac{\lambda}{\mu_1}+\left(\frac{\lambda}{\mu_1}\right)^2+\cdots+\left(\frac{\lambda}{\mu_1}\right)^N>\frac{r_+}{r_-} 
	\label{stability_reservoir}
\end{align}
(this can be derived by setting the mean drift, i.e. the depletion rate times the non-idle probability of the server minus the fill rate times the probability of the server being idle, to be negative).

Note that the stability condition here depends only on the service rate $\mu_1$, and $\mu_2$ can even be set to zero. However, as we will see next, in that case the server may never be idle, and the AoI may grow to infinity.

We can apply Lemma~\ref{lemma:finite_buff} to derive a closed-form expression for the peak AoI in an M/M/1/1 system (no waiting space, an update that comes to a busy server is rejected). 

The system of linear equations (\ref{sys_eq_fin_buff}) becomes
\begin{equation}
	\begin{aligned}
		\mu_1 y_1&=(\lambda+r_+ \xi_0^{(1)})y_0\;, \\
		y_0&=p_0\; \\
		\mu_2 p_1&=\lambda p_0 - (\mu_1-\mu_2)y_1\;, \\
		p_o+p_1&=1\;,
	\end{aligned}
\end{equation}
where $\xi_0^{(1)}=\displaystyle \frac{\mu_1}{r_-} - \frac{\lambda}{r_+}$ from (\ref{P_1(x)}). This is always negative, since for stability of the fluid reservoir we must have 
$\displaystyle \frac{\lambda}{\mu_1}>\frac{r_+}{r_-}$ (from (\ref{stability_reservoir})).

Solving the system of linear equations, we obtain
\begin{subequations}
	\begin{align}
		p_0&=\frac{\mu_1\mu_2}{\mu_1\mu_2+\frac{r_+}{r_-}(\mu_1\mu_2-\mu_1^2)+\lambda\mu_1}\;,\\
		p_1&=\frac{\frac{r_+}{r_-}(\mu_1\mu_2-\mu_1^2)+\lambda\mu_1}{\mu_1\mu_2+\frac{r_+}{r_-}(\mu_1\mu_2-\mu_1^2)+\lambda\mu_1}\;.
	\end{align}
\end{subequations}
$p_1$ is also the mean number of update packets in the server. Therefore, from Little's law, the sojourn time is $\mathbb{E}[S]=p_1/\lambda p_0$ and
\begin{align}
	\mathbb{E}[\text{AoI}_{peak,1}]&=\mathbb{E}[A]+\mathbb{E}[S] \nonumber\\
	&=\frac{1}{\lambda p_0}+\frac{\frac{r_+}{r_-}(\mu_1\mu_2-\mu_1^2)+\lambda\mu_1}{\lambda\mu_1\mu_2}\nonumber\\
	&=\frac{2\frac{r_+}{r_-}(\mu_1\mu_2-\mu_1^2)+2\lambda\mu_1+\mu_1\mu_2}{\lambda\mu_1\mu_2}\;.
	\label{peak_aoi_N1}
\end{align}

Next we present numerical results for the case of finite waiting room in the M/M/1/N queue with $N=1$ (1 packet in service, no waiting) and $N=2$ (1 packet in service, at most one waiting). We expect that having a minimal waiting room, or no waiting room at all, will significantly decrease age, as has been shown previously for the unregulated queue (see \cite{costa2014age,kesidis2020distribution}).

For the sake of comparison, we also calculate the mean peak age for the infinite buffer queue. From (\ref{exp_soj_time}), this equals:
\begin{equation}
	\mathbb{E}[\text{AoI}_{peak,\infty}]=\frac{1}{\lambda}+\frac{\zeta}{\lambda\sigma^{-1}(1-\sigma)}+\frac{1-\zeta}{\mu_2-\lambda}\;,
\end{equation}
where $\lambda$ is the arrival rate and $\zeta$, $\sigma$ are as defined in (\ref{zeta}), (\ref{sigma}), respectively.

Plots of the mean peak age are shown in Fig.~\ref{plots_MM11}. 
Overall, solid lines show results for the infinite buffer queue, dashed lines for $N=1$ and dotted lines for $N=2$.
The same configurations are selected for the ($r_+$, $r_-$) values, as in the performance evaluation in Section~\ref{sec:inf-inf}. For each configuration, we vary the arrival rate $\lambda$, the service rate $\mu_2$, and the buffer size $N$. In all cases we set $\mu_1=1$. 

The range of $\lambda$ values is determined by the stability constraint (\ref{stability_condition}) in the infinite buffer case, and (\ref{stability_reservoir}) in the finite buffer case. Having $N=2$ allows lower arrival rates compared to the $N=1$ case, since there are more packets available to consume energy (the difference can be seen more clearly in Fig.~\ref{plot_mm12_1_2}, \ref{plot_mm12_2_3}). Note that in the infinite buffer case, $\mu_2$ is also bounded by the stability constraint (\ref{stability_condition}). In the finite case it is not bounded, but it can be seen from (\ref{peak_aoi_N1}) that the mean peak age grows to infinity as $\mu_2\to 0$. Results are also shown for the unregulated queue ($\mu_1=\mu_2=1$).
\begin{figure*}[!tb]  
	\begin{subfigure}{0.5\textwidth}
		\includegraphics[width=\linewidth]{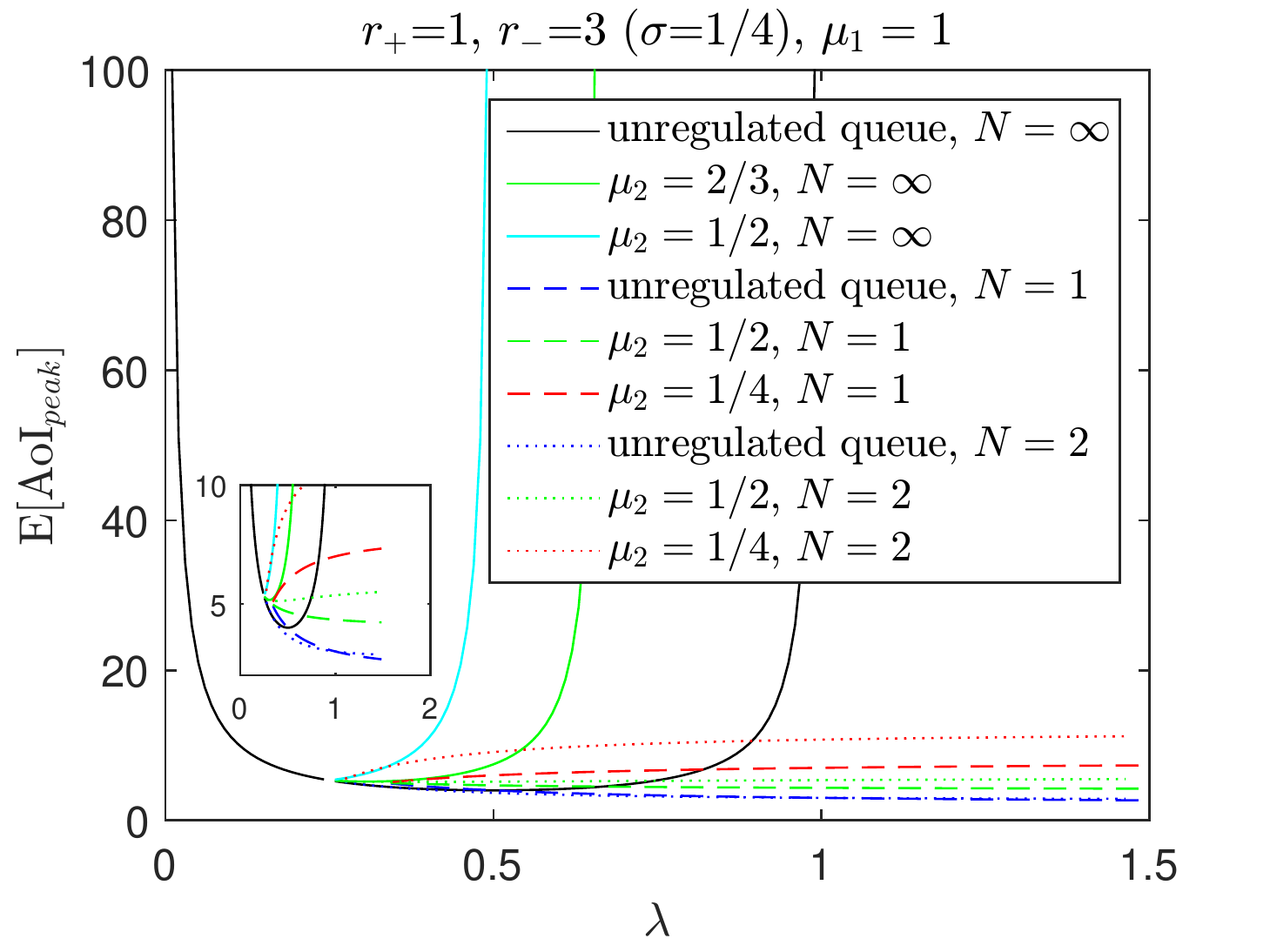}
		\caption{}
		\label{plot_mm12_1_4}
	\end{subfigure}
	\hfill 
	\begin{subfigure}{0.5\textwidth}
		\includegraphics[width=\linewidth]{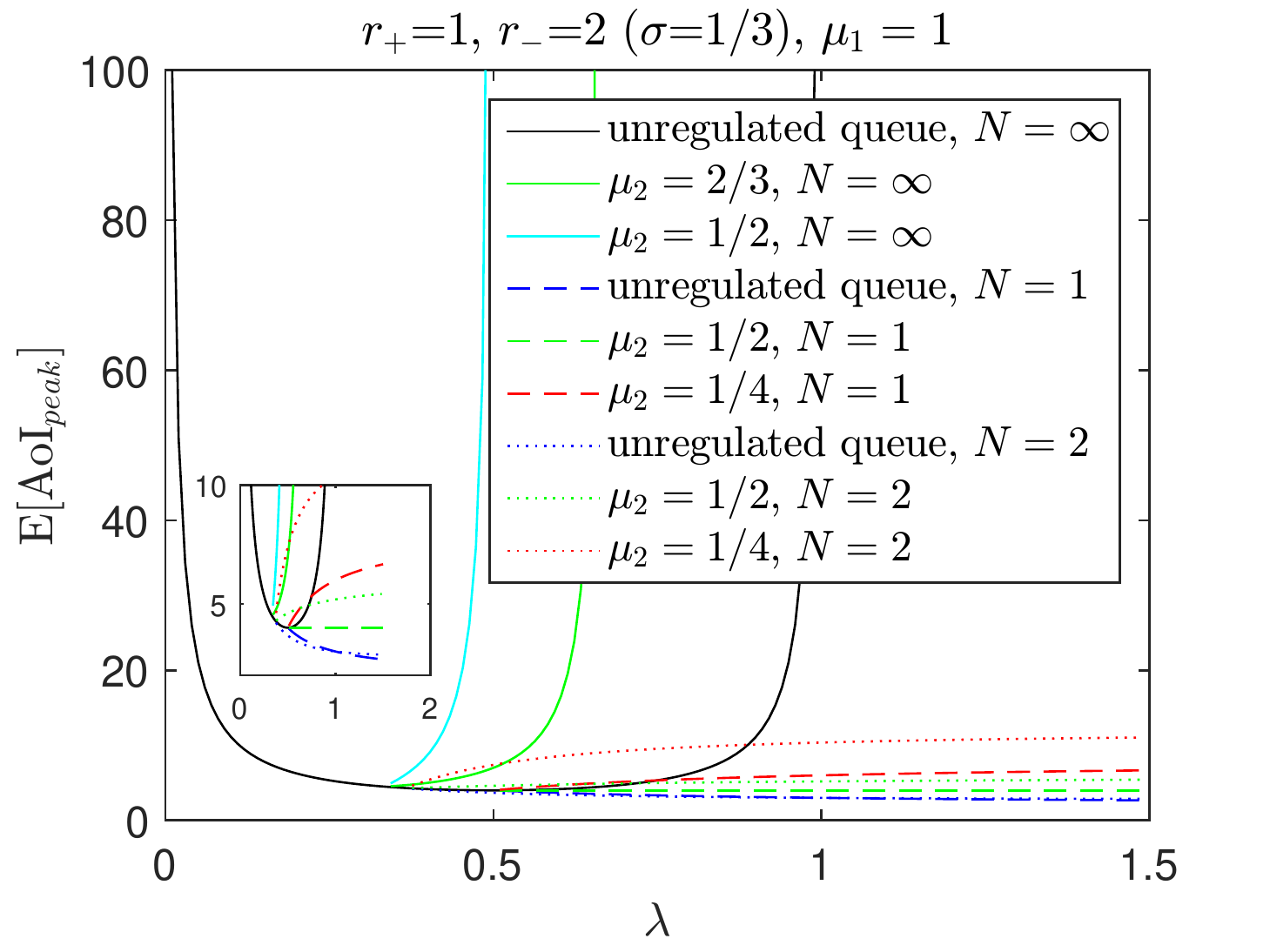}
		\caption{}
		\label{plot_mm12_1_3}
	\end{subfigure}
	\begin{subfigure}{0.5\textwidth}
		\includegraphics[width=\linewidth]{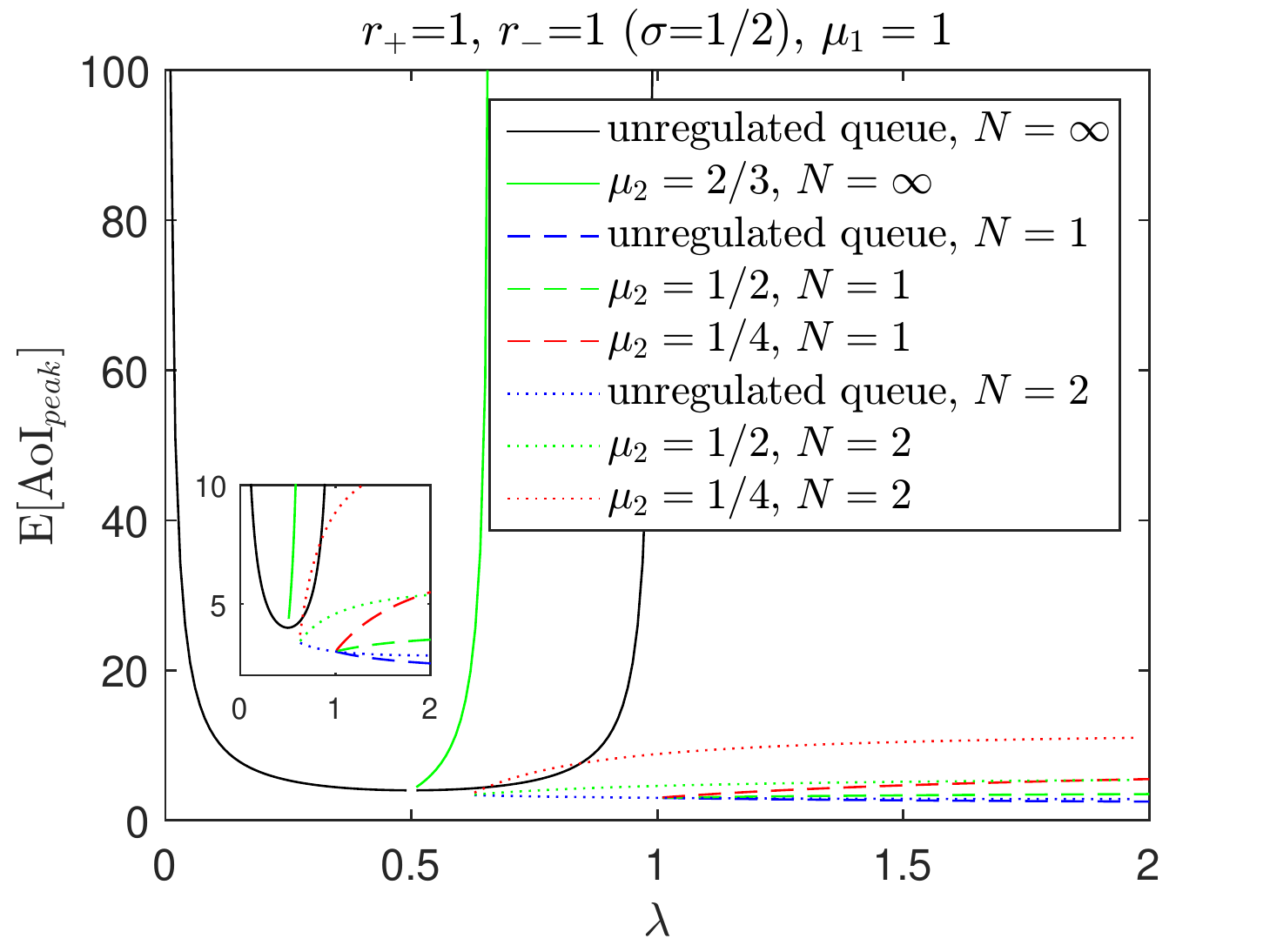}
		\caption{}
		\label{plot_mm12_1_2}
	\end{subfigure}
	\hfill 
	\begin{subfigure}{0.5\textwidth}
		\includegraphics[width=\linewidth]{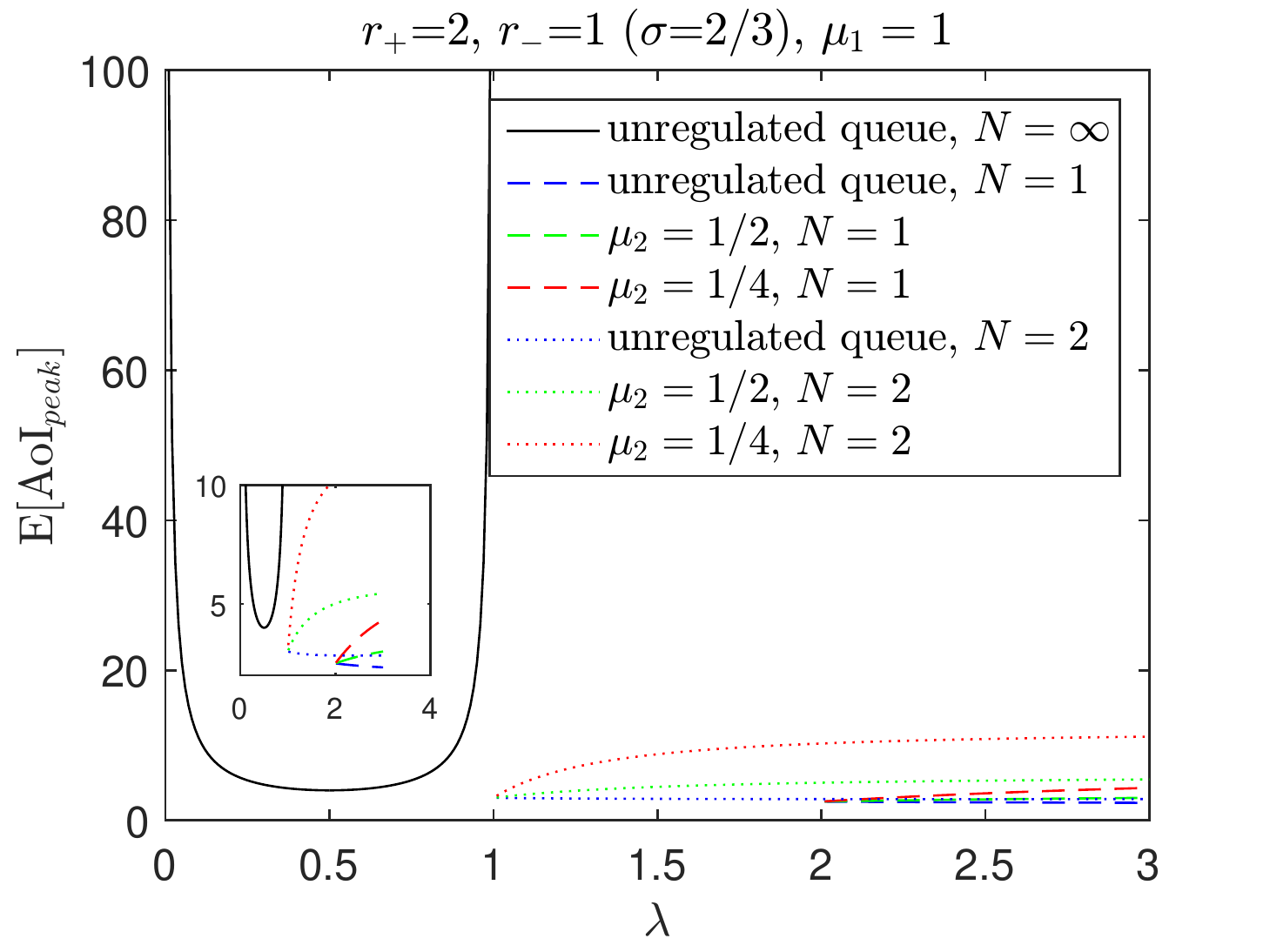}
		\caption{}
		\label{plot_mm12_2_3}
	\end{subfigure}
	\caption{$\mathbb{E}[\text{AoI}_{peak}]$ for finite waiting room and infinite reservoir} 
	\label{plots_MM11}
\end{figure*}

The mean peak AoI behaves similarly to the mean AoI metric for the infinite buffer cases, abruptly increasing after a certain arrival rate threshold, which primarily depends on the value of $\mu_2$.
We observe that the mean peak AoI is always smaller for a finite buffer, except for very small values of the arrival rate (see the inset plots). This was also noted in \cite{costa2014age} (for the mean peak AoI) and \cite{kesidis2020distribution} (for the mean AoI).

In the finite buffer cases, when the arrival rate increases, the behavior depends on the value of $\mu_2$: for relatively high $\mu_2$ values, the mean peak AoI slowly decreases with $\lambda$, while it increases for relatively lower $\mu_2$. Intuitively, the increase happens because the large number of arrivals depletes the energy reservoir, resulting in lower service (and hence lower refresh) rate. 

We also note that setting $N=2$ increases the mean peak AoI, except for small values of $\lambda$ in the unregulated cases and energy-poor regimes (see Fig.~\ref{plot_mm12_1_4}, \ref{plot_mm12_1_3}). However, in all regulated cases with $\mu_2<\mu_1$, adding waiting room did not help at all. 
\section{Infinite Waiting Room and Finite Reservoir}
\label{sec:inf-fin}
In the case of an infinite waiting room and finite reservoir, an approximate model was developed in \cite{adan1998analysis} based on exponentially distributed credit quanta gathered during idle periods of the server, whose size, when kept small enough, accumulates up to the reservoir capacity. The exponential assumption is necessary to model the system as a Markov chain (in busy periods) and then to derive the overall stationary distribution. The analysis leads to a complicated system of multi-variate recursive equations, which for a sufficiently good approximation requires a very large number of equations.

In Table~\ref{table1} we present indicative results for the mean peak AoI for specific parameter settings. The size of the reservoir is denoted by $D$. For finite $D$, the mean peak age $\mathbb{E}[\text{AoI}_{peak}^D]$ is calculated from the value of the sojourn time derived by simulation in \cite{adan1998analysis}, adding the mean interarrival time $1/\lambda$. The results are also compared to the case of an infinite reservoir ($D=\infty$). The constraint (\ref{stability_condition}) must be satisfied for the infinite reservoir, while for a finite reservoir the only constraint is $\lambda/\mu_2<1$.
\begin{table}[htbp]
	\begin{center}
		\caption{Mean peak AoI for different reservoir capacities ($\lambda=1$, $r_{+}=1$) \label{table1}}
		\begin{tabular}{c c c c c}
			\hline
			& & &  & \\[-6pt]
			$\mu_1$	& $\mu_2$  & $r_{-}$ & $D$ & $\mathbb{E}[\text{AoI}_{peak}^D]$ \\[2pt]
			\hline
			& & &  & \\[-6pt]
			\multirow[t]{4}{*}{2}	& \multirow[t]{4}{*}{1.5} & \multirow[t]{4}{*}{2} & 1 & 2.887 \\
			& & & 2 & 2.826 \\
			& & & 5 & 2.744 \\
			& & & $\infty$ & 2.700 \\
			\multirow[t]{4}{*}{1.5}	& \multirow[t]{4}{*}{1.1} & \multirow[t]{4}{*}{1} & 2 & 10.507 \\
			& & & 3 & 10.373 \\
			& & & 5 & 10.162 \\
			& & & $\infty$ & 9.857 \\
			\multirow[t]{4}{*}{1.5}	& \multirow[t]{4}{*}{1.1} & \multirow[t]{4}{*}{2} & 2 & 10.789 \\
			& & & 3 & 10.746 \\
			& & & 5 & 10.693 \\
			& & & $\infty$ & 10.667 \\
			\hline
		\end{tabular} 		
	\end{center}
\end{table}

As was anticipated, decreasing the capacity of the reservoir increases the mean peak age (as less energy is available to serve incoming updates). However, differences are not that large; from the smallest-size capacity, up to the case of an infinite reservoir, differences in age are less than one time unit. The only thing that matters significantly is the load in the server, and the infinite reservoir model seems to be a relatively close lower bound to the mean peak age for finite capacity.
\section{Conclusions and Open Issues}
\label{sec:conclusions}
Summarizing, we have studied the AoI performance of update packets through a single-server queue, regulated by a fluid reservoir.
The major conclusions that can be drawn from this research are as follows: 
\begin{itemize}
	\item Decreasing the reservoir capacity increases age, as less energy is available to serve incoming updates. However, the increase is not very high, and the infinite reservoir model could also be used to study energy-constrained devices.
	\item For an infinite transmitter buffer, the energy constraints require to decrease the update packet arrival rate even further than in the case of an unregulated queue, in order to achieve an optimal age value. But, what is more important for decreasing age, is keeping the transmitter buffer small.
	\item In the case of a small buffer, a higher update packet arrival rate is only helpful under energy-rich regimes, where the reservoir is not frequently depleted, or when the service rate upon depletion can remain relatively high. A similar conclusion is in \cite{arafa2019age}, where the authors stated that the optimal update rate increases with the amount of energy available.
	\item Under energy-poor regimes with a low service rate upon depletion, a high update rate is not recommended, as it depletes the energy reservoir quickly, resulting in lower refresh rates, and subsequently in higher age values. In this case, besides maintaining a low update rate, the best option is not to have waiting room at all, and discard update packets that arrive when the server is busy. 
\end{itemize}

The analysis and performance evaluation of the case where both the transmitter buffer and the energy reservoir are finite remains an open issue. Other issues that could be investigated are the performance for other service disciplines, such as LCFS with and without service preemption, as well as the (much more difficult) cases of multiple classes of update packets, or of a network of many energy-harvesting nodes. Finally, the extension to other (or more general) arrival and service distributions would further enhance the model's applicability to realistic environments.
\bibliographystyle{unsrt}
\bibliography{aoi_energy}

\begin{thebibliography}{10}

\bibitem{kaul2012real}
Sanjit Kaul, Roy Yates, and Marco Gruteser.
\newblock Real-time status: How often should one update?
\newblock In {\em 2012 Proceedings IEEE INFOCOM}, pages 2731--2735. IEEE, 2012.

\bibitem{costa2014age}
Maice Costa, Marian Codreanu, and Anthony Ephremides.
\newblock Age of information with packet management.
\newblock In {\em 2014 IEEE International Symposium on Information Theory},
  pages 1583--1587. IEEE, 2014.

\bibitem{kesidis2020distribution}
George Kesidis, Takis Konstantopoulos, and Michael~A Zazanis.
\newblock The distribution of age-of-information performance measures for
  message processing systems.
\newblock {\em Queueing Systems}, 95:203--250, 2020.

\bibitem{sun2017update}
Yin Sun, Elif Uysal-Biyikoglu, Roy~D Yates, C~Emre Koksal, and Ness~B Shroff.
\newblock Update or wait: How to keep your data fresh.
\newblock {\em IEEE Transactions on Information Theory}, 63(11):7492--7508,
  2017.

\bibitem{yates2020age}
Roy~D Yates, Yin Sun, D~Richard Brown~III, Sanjit~K Kaul, Eytan Modiano, and
  Sennur Ulukus.
\newblock Age of information: An introduction and survey.
\newblock {\em arXiv preprint arXiv:2007.08564}, 2020.

\bibitem{adan1998analysis}
Ivo~JBF Adan, Erik~A van Doorn, JAC Resing, and Werner~RW Scheinhardt.
\newblock Analysis of a single-server queue interacting with a fluid reservoir.
\newblock {\em Queueing Systems}, 29(2):313--336, 1998.

\bibitem{wu2017optimal}
Xianwen Wu, Jing Yang, and Jingxian Wu.
\newblock Optimal status update for age of information minimization with an
  energy harvesting source.
\newblock {\em IEEE Transactions on Green Communications and Networking},
  2(1):193--204, 2017.

\bibitem{arafa2019age}
Ahmed Arafa, Jing Yang, Sennur Ulukus, and H~Vincent Poor.
\newblock Age-minimal transmission for energy harvesting sensors with finite
  batteries: Online policies.
\newblock {\em IEEE Transactions on Information Theory}, 66(1):534--556, 2019.

\bibitem{farazi2018average}
Shahab Farazi, Andrew~G Klein, and D~Richard Brown.
\newblock Average age of information for status update systems with an energy
  harvesting server.
\newblock In {\em IEEE INFOCOM 2018-IEEE Conference on Computer Communications
  Workshops (INFOCOM WKSHPS)}, pages 112--117. IEEE, 2018.

\bibitem{farazi2018age}
Shahab Farazi, Andrew~G Klein, and D~Richard Brown.
\newblock Age of information in energy harvesting status update systems: When
  to preempt in service?
\newblock In {\em 2018 IEEE International Symposium on Information Theory
  (ISIT)}, pages 2436--2440. IEEE, 2018.

\bibitem{zheng2019closed}
Xi~Zheng, Sheng Zhou, Zhiyuan Jiang, and Zhisheng Niu.
\newblock Closed-form analysis of non-linear age of information in status
  updates with an energy harvesting transmitter.
\newblock {\em IEEE Transactions on Wireless Communications}, 18(8):4129--4142,
  2019.

\bibitem{bertsekas1992robert}
Dimitri Bertsekas and Robert Gallager.
\newblock {\em Data Networks (2nd Ed.)}.
\newblock Prentice-Hall, Inc., USA, 1992.

\bibitem{suresh2014efficient}
Harishankar Suresh, Anand Baskaran, KP~Sudharsan, U~Vignesh, T~Viveknath,
  P~Sivraj, and K~Vijith.
\newblock Efficient charging of battery and production of power from solar
  energy.
\newblock In {\em 2014 International Conference on Embedded Systems (ICES)},
  pages 231--237. IEEE, 2014.

\bibitem{boico2007solar}
Florent Boico, Brad Lehman, and Khalil Shujaee.
\newblock Solar battery chargers for {N}i{MH} batteries.
\newblock {\em IEEE Transactions on Power Electronics}, 22(5):1600--1609, 2007.

\end{thebibliography}
\end{document}